\documentclass[a4paper,UKenglish,cleveref,autoref]{lipics-v2019}

\usepackage[sort,numbers]{natbib}
\usepackage{multirow}

\title{Scalable and Jointly Differentially Private Packing}

\author  {Zhiyi Huang} {The University of Hong Kong }{zhiyi@cs.hku.hk}{}{This work is supported in by a RGC grant HKU17203717E.}
\author {Xue Zhu }{The University of Hong Kong }{xuezhu26@hku.hk}{}{} 

\authorrunning{Z.\ Huang and X.\ Zhu}

\Copyright{Zhiyi Huang and Xue Zhu}



\usepackage{mathtools}
\usepackage{amsthm}
\usepackage{bm}
\usepackage{xcolor}
\usepackage{setspace}

\newcommand{\zhiyi}[1]{\textcolor{blue}{(Zhiyi: #1)}}
\newcommand{\xue}[1]{\textcolor{magenta}{(Xue: #1)}}
\renewcommand{\vec}[1]{#1}

\newcommand{\E}{\mathbf{E}}
\newcommand{\Var}{\mathbf{Var}}
\renewcommand{\Pr}{\mathbf{Pr}}

\newcommand{\pmax}{p_{\max}}
\newcommand{\etasum}{\eta_{\mathrm{sum}}}
\newcommand{\kl}{D_{\mathrm{KL}}}
\newcommand{\maxdiv}{D_{\infty}}
\newcommand{\opt}{\mathrm{OPT}}
\newcommand{\alg}{\mathrm{ALG}}

\newcommand{\R}{\mathbf{R}}
\newcommand{\Lap}{\mathrm{Lap}}

\newcommand{\dist}{\mathcal{N}}

\usepackage{algorithm}
\usepackage{algcompatible}

\begin{document}

\maketitle

\begin{abstract}
    We introduce an $(\epsilon, \delta)$-jointly differentially private algorithm for packing problems.
    Our algorithm not only achieves the optimal trade-off between the privacy parameter $\epsilon$ and the minimum supply requirement (up to logarithmic factors), but is also scalable in the sense that the running time is linear in the number of agents $n$.
    Previous algorithms either run in cubic time in $n$, or require a minimum supply per resource that is $\sqrt{n}$ times larger than the best possible. 
\end{abstract}

\section{Introduction}
\label{sec:intro}

Suppose a trusted principal has $b$ copies of some privacy-sensitive good, and there are $n$ agents interested in getting a copy of it.
Each agent has some value for receiving a copy of the good.
The principal would like to choose a subset of up to $b$ agents to receive the good so that sum of their values is maximized.\footnote{This is also known as social welfare maximization in the literature of mechanism design.}
However, one of the agents, Alice, gets paranoid that the others may be able to learn a lot of information about her value for the sensitive good.
In particular, here is a hypothetical scenario that Alice worries about.
Suppose the principal simply allocate to the $k$ agents with the largest values, with Alice being one of them.
Then, all the other $n-1$ agents may exchange information and figure out that only $b-1$ of them get a copy and, hence, Alice must also get one.
Further suppose that the $b$-th highest value among them is, say, \$1,000; 
they would also learn that Alice's value for the sensitive good is at least \$1,000.
\emph{Is there an allocation algorithm that addresses Alice's concerns without losing too much in the objective?}

This problem has been studied in a series of works in the last few years~\cite{HsuHRRW/2016/SICOMP, HsuHRW/2016/SODA, HuangZ/2018/SODA}.
More broadly, let us consider a general packing problem with $m$ resources and $n$ agents. 
Each agent demands a bundle consists of a certain amount of each resource, and has a certain value for getting it.
The goal is to pick a subset of the agents such that granting them the corresponding bundles approximately maximizes to sum of the values, subject to the supply constraints of the resources, while protecting the privacy of any individual agent.
This line of works focus on a specific notation of privacy called joint differential privacy.
In a nutshell, it requires that for any individual agent, say, Alice in our example, an adversary shall not be able to learn more than a negligible amount of information about the agent's private information, i.e., her value and demands, from the allocations and prices of the other agents.

How can packing algorithms guarantee joint differential privacy?
At a high level, such algorithms leave some amount of supply of each resource unallocated in some meticulous and randomized way, so that even if someone knows the allocations and prices of all other agents as in Alice's hypothetical scenario, he will not be able to learn with certainty whether Alice gets a copy of the item, or her value for it.
Further, to ensure that the objective is approximately optimal even with the unallocated supplies, all jointly private algorithms require the supply of each resource to be sufficiently large.
Hence, the literature measures how good a jointly private algorithm is by the trade-off between the privacy level, quantified by a parameter $\epsilon > 0$, and the minimum supply requirement, subject to getting an additive $\alpha n$ approximation.
To this end, Huang and Zhu \cite{HuangZ/2018/SODA} show that a supply of $\frac{\sqrt{m}}{\epsilon \alpha}$ per resource, up to logarithmic factors, is both sufficient and necessary. 

Another important consideration is the running time of the algorithms.
This is particularly relevant for (jointly) differentially private algorithms, since they generally require the size of the dataset, i.e., $n$, to be sufficiently large to achieve good approximation in the objective.
We argue that practical (jointly) differentially private algorithms must be scalable in the sense by Teng \cite{Teng/2016/FTTCS}, i.e., the running time shall be quasi-linear in $n$ or better.
However, the aforementioned algorithm by Huang and Zhu~\cite{HuangZ/2018/SODA}, which achieves the optimal $\tilde{O} \big( \frac{\sqrt{m}}{\epsilon \alpha} \big)$ supply requirement, is not scalable, as its running time depends cubically in $n$.
Neither are the earlier algorithms by Hsu et al.~\cite{HsuHRRW/2016/SICOMP, HsuHRW/2016/SODA}.
Although Huang and Zhu \cite{HuangZ/2018/SODA} also propose an alternative algorithm that is scalable, it requires a much larger supply of $\tilde{O} \big( \frac{\sqrt{m n}}{\epsilon \alpha} \big)$ per resource. 
The following question is explicitly left open~\cite{HuangZ/2018/SODA}:

\begin{quote}
    \em
    Are there jointly differentially private algorithm that are scalable and at the same time only require a minimum supply of $\tilde{O} \big( \frac{\sqrt{m}}{\epsilon \alpha} \big)$ per resource?
\end{quote}

\subsection{Our Contributions}

We introduce a jointly differentially private packing algorithm that answers the above open question affirmatively.
The main theorem of this paper is the following:

\begin{theorem}
    \label{thm:scalable-optimal-jdp-packing}
    There is an $(\epsilon, \delta)$-jointly differentially private algorithm such that: 
    \begin{enumerate}
        \item it returns with high probability a feasible packing solution that is optimal up to an $\alpha n$ additive factor, provided that the supply per resource is at least $\tilde{O} \big( \frac{\sqrt{m}}{\alpha \epsilon} \big)$;
        \item it stops in $O(n)$ time, omitting dependence in other parameters, with high probability.
    \end{enumerate}
\end{theorem}

The algorithm follows the same high-level framework as the previous ones, which we summarize below.
It maintains for each resource a price (per unit of the resource), which can be viewed as a dual variable that Lagrangianizes the corresponding resource constraint of the packing problem.
We shall imagine that the prices are posted on a public billboard for everyone to see, including the agents and the adversary.
Given the current prices, each agent gets the bundle if and only if her value is higher than the total price of the bundle.
The agents' decisions induce a total demand on each resource, which can be viewed as a subgradient for the dual prices w.r.t.\ some dual objective.
The algorithm then increases the prices of the overdemanded resources and decreases those of the underdemended ones.
This process repeats for a certain number of rounds; the final allocation is obtained by averaging over all the rounds.
Since the agents' allocation is coordinated only through the prices, it suffices to ensure that the sequence of prices is privacy-preserving.
This is formulated as the billboard lemma by Hsu et al.~\cite{HsuHRRW/2016/SICOMP}.

What are the main differences between our algorithm and the existing ones?
The previous non-scalable algorithms are essentially noisy versions of some existing optimization algorithms, including gradient descent~\cite{HsuHRW/2016/SODA} and multiplicative weight update~\cite{HuangZ/2018/SODA}, run on the dual space with a \emph{fix step size}.
They are not scalable because (1) the number of rounds needed by such algorithm generally depends on how large each coordinate of the subgradient could be (a.k.a., the width of the problem), which is roughly $n$ in our problem, and (2) the time needed to compute the subgradient in each round is linear in $n$.

The scalable algorithm by Huang and Zhu~\cite{HuangZ/2018/SODA}, on the other hand, is a noisy version of the online multiplicative weight update algorithm~(e.g., \cite{AgrawalD/2014/SODA}), which use the demand of a single agent as an estimator of the overall demand in each round, iterating through all agents once in a random order.
However, it does not seem plausible to avoid having an extra $\sqrt{n}$ factor in the minimum supply requirement using this approach, as it is not only jointly differentially private, but also locally private,\footnote{This is not explicitly stated in Huang and Zhu \cite{HuangZ/2018/SODA}. Nonetheless, it follows straightforwardly from the definition of the algorithm.} in the sense that it can be implemented in a way such that the agents add noises themselves so that even the algorithm never accesses any non-private version of the data.
The extra $\sqrt{n}$ factor is ubiquitous in the literature of locally private algorithms~\cite{BassilyS/2015/STOC, DuchiJW/2013/FOCS}.

In contrast, our algorithm is a noisy version of the multiplicative weight update algorithm run on the dual space with \emph{different step sizes}, which are optimized according on the scale of the subgradient in each round.
Intuitively, it chooses a small step size when the scale of the subgradient is large, to avoid dramatic changes in prices, and a large step size when the scale of the subgradient is small, to ensure a good enough progress. 
The idea of using different step sizes is widely used in non-private packing algorithms to get width-independent running time (e.g.,~\cite{KoufogiannakisY/2014/Algorithmica}).
Despite being standard in non-private packing, a direct combination of it and how the existing approaches add noises to the subgradients lead to suboptimal minimum supply requirement and/or super-linear running time in $n$.
Instead, we need to further use different noise scales in different rounds that are tailored to the scales of the subgradients and, by induction, the corresponding step sizes.
In a round where the scale of the subgradient is large (respectively, small) and the step size is small (respectively, large), the algorithm adds noises at a larger (respectively, smaller) scale to the subgradient and, in some sense, uses up less (respectively, more) of the privacy budget.
Setting noise scales adaptively over time introduces several technical difficulties which we will address in details in the technical sections.
To get the results in Theorem~\ref{thm:scalable-optimal-jdp-packing}, the noise scale is inversely proportional to the square root of the step size.
Hence, our algorithm is more precisely characterized as a noisy version of the multiplicative weight update algorithm run on the dual space with \emph{different step sizes and noise scales}, both of which are optimized according on the scale of the subgradient in each round.
Table~\ref{tab:summary} provides a brief comparison of our algorithm and the existing ones.

\begin{table}
    \renewcommand{\arraystretch}{1.3}
    \begin{tabular}{|c|c|c|c|}
        \hline
        Reference & Algorithm & Min Supply & Running Time (in $n$)
        \\
        \hline
        Hsu et al. \cite{HsuHRW/2016/SODA} & Dual GD & $\tilde{O} \big( \frac{m}{\alpha\epsilon} \big)$ & $O(n^3)$
        \\
        \hline
        \multirow{2}{*}{ Huang and Zhu \cite{HuangZ/2018/SODA}} & Dual MWU (fixed step size) & $\tilde{O} \big( \frac{\sqrt{m}}{\alpha\epsilon} \big)$ & $O(n^3)$
        \\
        \cline{2-4}
        & Dual Online MWU & $\tilde{O} \big( \frac{\sqrt{mn}}{\alpha\epsilon} \big)$ & $O(n)$
        \\
        \hline 
        \textbf{This paper} & Dual MWU (different step sizes) & $\tilde{O} \big( \frac{\sqrt{m}}{\alpha\epsilon} \big)$ & $O(n)$ 
        \\
        \hline
    \end{tabular}
    \medskip
    \caption{A comparison of the algorithm in this paper and those in previous works. $n$ and $m$ denote the number of agents and number of resources respectively. $\epsilon$ and $\alpha$ quantify the privacy and approximation guarantees respectively.}
    \label{tab:summary}
\end{table}
\subsection{Related Work}
The notion of differential privacy is introduced by Dwork et al.~\cite{DworkMNS/2006/TCC}. 
It has evolved through a long line of works to become a standard notion of privacy in theoretical computer science.
See Dwork and Roth~\cite{DworkR/2014/FTTCS} for a textbook introduction.
A particularly related line of works study differentially private algorithms for combinatorial optimization problems and mathematical programs.
McSherry and Talwar~\cite{McSherryT/2007/FOCS} introduce a generic (yet computationally inefficient) method called the exponential mechanism for privately solving optimization problems whose feasible set of output is independent on the dataset (e.g., $k$-means clustering). 
It is not applicable to the packing problem considered in this paper, since the set of feasible allocations crucially rely on the dataset. 
Hsu et al.~\cite{HsuRRU/2014/ICALP} study what linear programs can be solved in a differentially private manner.
Unfortunately, packing linear programs are not among the solvable ones~\cite{HsuHRRW/2016/SICOMP}.

Subsequently, Kearns et al.~\cite{KearnsPRU/2014/ITCS} introduce a relaxed notion called joint differential privacy.
It is still strong enough to provide provable privacy guarantees, but is also flexible enough to allow positive results for problems that cannot be solved under the original notion of differential privacy.
This relaxed notion is widely used not only in resource allocation problems~\cite{HsuHRRW/2016/SICOMP, HsuHRW/2016/SODA, HuangZ/2018/SODA} such as the packing problem considered in this paper, but also in coordinating large games~\cite{RogersR/2014/EC, CummingsKRW/2015/WINE, RogersRUW/2015/EC, KearnsPRU/2014/ITCS, CummingsLRRW/2016/ITCS}, privacy-presering learning~\cite{CummingsIL/2015/COLT, ShariffS/2018/NeurIPS}, privacy-preserving surveys~\cite{GhoshLRS/2014/EC}, privacy-preserving prediction markets~\cite{CummingsPV/2016/EC}, etc.

Recently, jointly differentially private resource allocation algorithms find further applications in regularizing strategic behaviors in the problem of learning reserve prices online in strategic environments~\cite{LiuHW/2018/NeurIPS}.

\section{Model}
\label{sec:model}

For any positive integer $\ell$, let $[\ell]$ denote the set of integers between $1$ and $\ell$, i.e., $\{1,2,\dots,\ell\}$.

\paragraph*{Packing}
Consider a packing problem with $n$ agents and $m$ resources. 
Each agent $i \in [n]$ demands a bundle of resources; let $a_{ij}$ denote her demand for each resource $j \in [m]$.
Further, agent $i$ has value $v_i$ for getting the bundle, and $0$ for not getting it.
We assume that $a_{ij}$'s and $v_i$'s are bounded between $0$ and $1$, which is standard in the literature of differential privacy.
For any agent $i \in [n]$, the demands $a_{ij}$'s and the value $v_i$ are her private data. 
Let $U = \{ (v, a_1, a_2, \dots, a_m) \in [0, 1]^{m+1} \}$ denote the data universe.
Let $\mathcal{D} \in U^n$ denote a dataset of $n$ agents. 
For any resource $j \in [m]$, let $b_j$ denote its supply.
The goal is then to choose a subset of the agents who get their demanded bundles, such that the sum of the values of the chosen agents is maximized, subject to that the total demand on each resource does not exceed the corresponding supply.

We remark that we can make two simplifying assumptions because the focal point of the jointly differentially private packing problem lies in whether the minimum supply is sufficiently large. 
First, we may assume without loss of generality (wlog) that the supplies of all resources are equal; otherwise, we may rescale the larger ones down to be equal to the minimum one. 
Second, we may focus on fractional solutions wlog; any fractional solution can be converted into an integral one using independent rounding with essentially the same objective and total demands on the resources, since we are in the large supply regime.

Therefore, the problem can be formulated as the following packing linear program:
\begin{align*}
\textrm{maximize} ~~ & \textstyle \sum_{i\in[n]} v_i x_i & &\\
\textrm{subject to} ~ & \textstyle \sum_{i\in[n]} a_{ij}x_i \le b & & \forall j \in [m] \\
& \textstyle 0\le x_{i} \le 1 & & \forall i \in [n]
\end{align*}

\paragraph*{Differential Privacy and Joint Differential Privacy}

Next, we formally define differential privacy and joint differential privacy with respect to the packing problem.
Two datasets $\mathcal{D}, \mathcal{D}' \in U^n$ are $i$-neighbors if they differ only in the data of the $i$-th agent, that is, if $\mathcal{D}_j = \mathcal{D}'_j$ for all $j \neq i$. 
We simply say that they are neighbors if they are $i$-neighbors for some $i \in [n]$.
Further, let $X_i = [0, 1]$ denote the set of feasible decision to each agent $i \in [n]$.
Let $X = X_1 \times X_2 \times \dots \times X_n$ denote the set of feasible outcomes, ignoring the supply constraints. 
The notion of differential privacy by Dwork et al.~\cite{DworkMNS/2006/TCC} requires that the allocation of all agents is chosen from similar distributions for any neighboring datasets in the following sense:

\begin{definition}[Differential Privacy] 
    A mechanism $\mathcal{M}: U^n \mapsto X$ is $(\epsilon,\delta)$-differentially private if for any neighbors $\mathcal{D}, \mathcal{D}'\in U^n$, and any subset of feasible allocations $S \subseteq X$:
    \[
	    \Pr \big[ \mathcal{M}(\mathcal{D}) \in S \big] \le \exp(\epsilon)\cdot \Pr \big[ \mathcal{M}(\mathcal{D}') \in S \big] + \delta 
        ~.
    \]
\end{definition}

The notion of joint differential privacy by Kearns et al. \cite{KearnsPRU/2014/ITCS}, on the other hand, allows the allocation to each agent to depend non-privately on her own data, so long the allocation to the other agents does not.
More precisely, the notion is defined as follows:

\begin{definition}[Joint Differential Privacy] 
    A mechanism $\mathcal{M}: U^n \mapsto X$ is $(\epsilon,\delta)$-jointly differentially private if for any $i \in [n]$, any $i$-neighbors $\mathcal{D}, \mathcal{D}' \in U^n$, and any subset of feasible allocations to agents other than $i$, $S_{-i} \subseteq X_{-i}$:
	\[
        \Pr \big[ \mathcal{M}(\mathcal{D})_{-i} \in S_{-i} \big] \leq \exp(\epsilon) \cdot \Pr \big[ \mathcal{M}(\mathcal{D}')_{-i} \in S_{-i} \big] + \delta 
        ~.
    \]
\end{definition}

The technical connections between the two notions are best explained by the following billboard lemma by Hsu et al. \cite{HsuHRRW/2016/SICOMP}.

\begin{lemma}[Billboard Lemma]
    \label{lem:billboard}
    Suppose $\mathcal{M} : U^n \mapsto Y$ is $(\epsilon, \delta)$-differentially private. 
    Then, for any collection of functions $f_i : U \times Y \mapsto X_i$, $i \in [n]$, the mechanism $\mathcal{M}'$ that allocates to each agent $i$ with $f_i \big(\mathcal{D}_i, \mathcal{M}(\mathcal{D})\big)$ is $(\epsilon, \delta)$-jointly differentially private.
\end{lemma}

\section{Technical Preliminaries}
\label{sec:prelim}

\subsection{Lagrangian}

The Lagrangian of the packing linear program is:
\[
    \textstyle \max_{x \in X}\min_{p \in [0,\infty)^m} \sum_{i\in[n]}v_ix_i-\sum_{j\in[m]}p_j\big( \sum_{i\in[n]}a_{ij}x_i-b_j \big)
    ~.
\]

Let $L(x,p)$ denote the partial Lagrangian objective, that is:
\begin{align*}
L(x,p)
& \textstyle 
= \sum_{i\in[n]} v_i x_i - \sum_{j\in[m]} \big( \sum_{i\in[n]} a_{ij} x_i-b \big) p_j \\
&
 \textstyle  
= \sum_{j\in[m]} b p_j + \sum_{i\in[n]} \big( v_i-\sum_{j\in[m]} a_{ij} p_j \big) x_i
~.
\end{align*}

Let $D(p)=\max_{x\in[X]} L(x,p)$ denote the dual objective. 
We shall interpret $p_j$ as the unit price of resource $j$ for any $j\in[m]$. Given any prices $p$, an optimal solution $x^*(p)$ of the optimization problem $\max_{x\in[X]} L(x,p)$ is defined as, for any $i \in [n]$:
\[
    x^{*}_{i}(p) = 
    \begin{cases} 
        1, & \mbox{if $v_i - \sum_{j\in[m]} a_{ij} p_j \geq 0$;} \\
        0, & \mbox{otherwise.}
    \end{cases}
\]
%

We have the following envelope theorem (see, e.g., Afriat~\cite{Afriat/1971/SIAP}).
\begin{lemma}[Envolope Theorem]
	\label{lem:envolpe}
    Given any prices $p$, the derivatives of the Lagrangian objective w.r.t.\ the prices, i.e., $\nabla_p L(x,p)$, is a sub-gradient of $D(p)$ when $x = x^*(p)$.
\end{lemma}

\subsection{Truncated Laplacian Distributions}

The Laplacian distribution, given mean $\mu \in R$ and scale parameter $b > 0$, is a continuous distribution defined on $(-\infty, +\infty)$ such that the probability density of any $x \in \R$ is:
\[
    \frac{1}{2b} \exp \left( - \frac{|x - \mu|}{b} \right)
    ~.
\]
Let $\Lap(\mu, b)$ denote this distribution.
It has mean $\mu$ and variance $2b^2$.

Further, we will consider the truncated Laplacian distribution with support $[\mu-1+\alpha, \mu+1-\alpha]$, denoted as $\Lap_{1-\alpha}(\mu, b)$.
The probability density of any $x$ in the interval is proportional to that of $\Lap(\mu, b)$;
the density is $0$ if it is outside the interval.
$\Lap_{1-\alpha}(\mu, b)$ also has mean $\mu$, and variance $O(b^2)$.
(See below for a formal statement.)

Given any mean $\mu \in [-\alpha, \alpha]$ and target standard deviation (up to a constant factor) $0 < \sigma \le \alpha$, let $\dist(\mu, \sigma) = \Lap_{1-\alpha}(\mu, \sigma)$.
We use this notation to emphasize that the family of noise distributions, $\dist(\mu, \sigma)$'s, can be replaced by distributions other than the truncated Laplacian distributions, as long as they satisfy the following properties.
The proofs are deferred to Appendix~\ref{sec:miss-proof-truncated}.

\begin{lemma}
    \label{lem:noise-distributions-statistics}
    The noise distributions $\dist(\mu, \sigma)$'s satisfy that:
    \begin{enumerate}
        \item the mean of $\dist(\mu, \sigma)$ is $\mu$;
        \item the variance of $\dist(\mu, \sigma)$ is at most $2 \sigma^2$. 
    \end{enumerate}
\end{lemma}

\begin{lemma}
    \label{lem:noise-distributions-privacy}
    Suppose $-\alpha \le \mu_1, \mu_2 \le \alpha$, and $0 < \sigma_1, \sigma_2 \le \alpha$ satisfy that for some $0<\eta \leq \alpha$, $\sigma = \max \{\sigma_1, \sigma_2\}$, and $\delta > 0$: (1) $\big| \frac{1}{\sigma_1^2} - \frac{1}{\sigma_2^2} \big| \le \frac{\eta}{\alpha \sigma^2}$; (2) $\big| \mu_1 - \mu_2 \big| \le \eta$; and $\delta \le \frac{\eta \ln(2/\delta)}{\sigma}$.
    Then, for any $S \subseteq \R$, we have:
    \[
        \Pr_{z \sim \dist(\mu_1, \sigma_1)} \big[ z \in S \big] \le \exp\left(\frac{4 \eta \ln(2/\delta)}{\sigma}\right) \cdot \Pr_{z \sim \dist(\mu_2, \sigma_2)} \big[ z \in S \big] + \delta
        ~.
    \]
\end{lemma}

\section{Our Algorithm}
\label{sec:algo}

The algorithm follows a primal dual approach, running best response on the primal (i.e., allocation $\vec{x}$) and a noisy version of the multiplicative weight update (MWU) method on the dual (i.e., the prices $\vec{p}$).
Similar approaches are also used in the previous works that study jointly differentially private packing problem (e.g., \cite{HsuHRW/2016/SODA, HuangZ/2018/SODA}).
The new ingredients of our algorithm are the use of nonuniform step sizes as well as nonuniform noise scales, both of which are meticulously optimized to achieve both scalable running time and the optimal trade-off between privacy and the minimum supply requirement.
In contrast, all previous algorithms use a fix step size and a fix noise scale across different rounds.

\begin{algorithm}[t]
	\caption{Private Dual MWU with Optimized Step Sizes and Noise Scales}
	\label{alg:mw-optimized-step-sizes}
	\begin{algorithmic}[1]
        \STATEx \textbf{Input:~} 
        \begin{itemize}
            \item Dataset $\mathcal{D} \in U^n$, represented by $a_{ij}$'s and $v_i$'s; 
            \item Supply (per resource) $b$.
        \end{itemize}
        %
		%
        \STATEx \textbf{Assumptions:}
        \begin{itemize}
            \item $b \ge \tilde{O}\big(\tfrac{\sqrt{m}}{\alpha \epsilon}\big)$;
            \item We also assume $n \ge b$ as the problem is trivial otherwise.
        \end{itemize}
		\STATEx \textbf{Parameters:}
        \begin{itemize}
            \item Upper bound on dual prices $p_{\max} = \frac{2n}{b}$;
            \item Initial dual prices $\vec{p}^1$ such that $p^1_j = \frac{\pmax}{m+1}$ for any $j \in [m+1]$; 
	        \item Upper bound on the sum of step sizes $\etasum = \frac{\ln(m+1)}{\alpha b}$.
        \end{itemize}
        \FOR{$t = 1, 2, \dots,$ \textbf{until} $\sum_t \eta^t \geq \etasum$ (via a private counter)} 
            \STATE Let $\vec{x}^t = \vec{x}^*(\vec{p}^t)$, i.e., the best response to $\vec{p}^t$ from the primal viewpoint.
            \STATE Let $\nabla_j D(p^t) = b -\sum_{i \in [n]} a_{ij} x^t_i$ for all $j \in [m]$; let $\nabla_{m+1} D(p^t) = 0$.
            \STATE Let the \emph{step size} and the \emph{noise scale} be:
            \[
                \eta^t = \min \left\{ \frac{\alpha}{b}, \frac{\alpha}{\nabla_1 D(p^t)}, \frac{\alpha}{\nabla_2 D(p^t)}, \dots, \frac{\alpha}{\nabla_m D(p^t)} \right\}
                \quad,\quad
                \sigma^t = \frac{\sqrt{m\etasum \eta^t \ln (Tm/\delta) }}{\epsilon} 
                ~.
            \]
            \STATE Draw $\delta^t_j \sim \dist\big( \mu^t_j, \sigma^t \big)$ where $\mu^t_j = \eta^t \nabla_j D(p^t)$ for all $j \in [m]$; let $\delta^t_{m+1} = 0$.
            \STATE Let $\hat{p}^{t+1}_j=p^t_j \cdot \exp(-\delta^t_j)$ for all $j \in [m+1]$.
            \STATE Let $p^{t+1}$ be such that $p^{t+1}_j \propto \hat{p}^{t+1}_j$ and $\sum_{j \in [m+1]} p^{t+1}_j = \pmax$.
        \ENDFOR
        \STATE Let $T$ be the number of iterations in the for loop.
        \STATEx \textbf{Output:} $\bar{\vec{x}} = \tfrac{1}{\etasum} \sum_{t = 1}^T \eta^t \vec{x}^{t}$, where agent $i$ observes $\bar{x}_i$.
	\end{algorithmic}
\end{algorithm}

See Algorithm~\ref{alg:mw-optimized-step-sizes} for an exposition by pseudocode.

We now describe the algorithm in more details.
Let us first explain the basic dynamic of the dual MWU algorithm.
For technical reasons, we add a dummy resource with $0$ supply, and $0$ demands from all agents, and assume that the dual prices sum to a fix and sufficiently large number $\pmax$.
Starting from some initial guess $\vec{p}^1$ of the prices, say, with $\pmax$ uniformly distributed among the $m+1$ coordinates, the MWU algorithm repeatedly calculates:
\[
    \vec{x}^t = \vec{x}^* \big( \vec{p}^t \big) 
    \quad;\quad
    p^{t+1}_j \propto p^t_j \cdot \exp \left( - \eta^t \nabla_j D\big(p^t\big) \right) ~,~ \forall j \in [m+1] 
\]
where $\eta^t > 0$ is the step size of round $t$.
That is, the allocation $\vec{x}^t$ is the best response to $\vec{p}^t$, which induces a subgradient of the dual objective at $\vec{p}^t$ by Lemma~\ref{lem:envolpe}. 
Then, from $\vec{p}^t$ to $\vec{p}^{t+1}$, each coordinate $j \in [m+1]$ decreases exponentially by an amount proportional to corresponding subgradient, before they are rescaled to sum to $\pmax$.
With appropriate step sizes, standard analysis shows that the weighted average allocation across different rounds, where the weight of each round is its step size, converges to an optimal allocation.\footnote{The prices also converge, although this is not relevant for our analysis.}

To obtain the desired privacy guarantee, our algorithm updates the prices with zero-mean noises added to the subgradients.
Below we discuss the choice of step sizes and noise scales.

\bigskip

\noindent
\textbf{Step Sizes:~}
%
Some standard choices of step sizes include uniform step sizes, i.e., $\eta^t = \eta$, which is used in previous works on jointly private packing algorithms~\cite{HsuHRW/2016/SODA, HuangZ/2018/SODA}, decreasing step sizes, e.g., $\eta^t = \tfrac{\eta}{t}$, and step sizes inversely proportional to the magnitude of the subgradient, which are what our algorithm uses (see, e.g., Koufogiannakis and Young \cite{KoufogiannakisY/2014/Algorithmica}, for an application in the non-private packing problem).
Intuitively, our choice of step sizes ensure that at least one coordinate will be updated by an $\eta$ amount, which in turns lower bounds the amount of progress made in each round.
Our algorithm has one caveat, however, as it further caps the step size by an upper bound, which is set to $\frac{\alpha}{b}$ for technical reasons, so that the noise added in any single round does not affect the result by too much. 

\bigskip
\noindent
\textbf{Noise Scales:~}
%
Not surprisingly, our first attempt is to add a uniform amount of noise to every round like the previous jointly private algorithms in the literature~\cite{HsuHRRW/2016/SICOMP, HsuHRW/2016/SODA, HuangZ/2018/SODA}.
This, however, either requires the minimum supply to be much larger than the best possible, or is not scalable. 
To see why, let us first consider an overly-simplified argument of why uniform noise scales work in the previous algorithms.
Suppose the algorithm takes $T$ rounds in total. 
By a standard composition theorem of differential privacy~(see, e.g., Dwork and Roth \cite{DworkR/2014/FTTCS}), adding noises at a uniform scale $\tilde{O} \big( \frac{\sqrt{T}}{\epsilon} \big)$ is sufficient for achieving $(\epsilon, \delta)$-joint differential privacy.
Then, if the algorithm uses a uniform step size $\eta$, by the standard concentration bound, the cumulative noise summing over $T$ rounds is roughly $\tilde{O} \big( \frac{\eta T}{\epsilon} \big)$.
After averaging, this is essentially a fixed amount of noise $\tilde{O} \big( \frac{1}{\epsilon} \big)$ independent of $T$ and $\eta$!

With non-uniform step sizes, however, this is no longer true.
As a thought experiment, suppose there are $T' \ll T$ rounds that have large step sizes, say, all equal to $\eta$; the rest of the rounds can be omitted due to negligible step sizes.
Note that the uniform noise scale, i.e., $\tilde{O} \big( \frac{\sqrt{T}}{\epsilon} \big)$, is still determined by the total number of rounds, repeating the above calculation gives that the amount of noise after averaging is $\tilde{O} \big( \frac{\sqrt{T}}{\epsilon \sqrt{T'}} \big) \gg \tilde{O} \big( \frac{1}{\epsilon} \big)$. 

The lesson we learned from this though experiment is that the algorithm must choose non-uniform noise scales: 
smaller noise scales for more important rounds that have larger step sizes;
and larger noise scales for less important rounds that have smaller step sizes.
More precisely, we optimize the noise scale in each round $t$ to be inversely proportional to the square root of the step size, i.e, $\sqrt{\eta^t}$; this is derived form a Cauchy-Schwarz inequality to balance different aspects of the analysis.
(Note that $\delta^t_j$'s in Algorithm~\ref{alg:mw-optimized-step-sizes} denote the noises added to the subgradients \emph{multiplied by the corresponding step size} and, hence, its scale is proportional to $\sqrt{\eta^t}$, rather than inversely proportional to it.)

\bigskip

\noindent
\textbf{Privacy-preserving Stopping Criteria:~}
%
Finally, note that the stopping criteria of $\sum_t \eta^t \ge \etasum$ must be implemented approximately in a privacy-preserving manner as well.
This can be done by maintaining $\sum_t \eta^t$ using a standard technique called private counter (e.g., Chan et al.~\cite{ChanSS/2010/ICALP}, Dwork et al.~\cite{ DworkNPR/2010/STOC}).
For simplicity of exposition, we will omit this standard component and analyze the algorithm assuming the stopping criteria is implemented exactly.

\section{Utility and Time Complexity}
\label{sec:utility}

This section sketches the analysis of utility guarantees provided by Algorithm~\ref{alg:mw-optimized-step-sizes} and its time complexity, under the assumption that the supply is sufficiently large, i.e., $b \ge \tilde{O} \big( \frac{\sqrt{m}}{\alpha \epsilon} \big)$.

\subsection{No-regret Lemma}

We first introduce a technical lemma that will serve as the overarching tool in the analysis of the algorithm's approximation guarantees in terms of the objective and constraint violations, and its running time.
It states that if we compare the Lagrangian objective achieved by the sequence of $x^t$'s and $p^t$'s computed in the algorithm, and what could have be achieved by replacing $p^t$'s with an arbitrary but fix $p$, the difference can be bounded. 
In other words, the dual price sequence has no regret in the terminology of online learning.

\begin{lemma}
	\label{lemma:l(xt,pt)-l(xt,p)}
	For any $p$ such that $\|p\|_1=p_{max}$,  with high probability, we have:
	\begin{equation*}
        \textstyle
        \sum_{t=1}^T \eta^t \big( L(x^t,p^t) - L(x^t,p) \big) 
        \leq 
        \kl(p\|p^1) + \etasum \cdot O \big( \alpha n \big) 
        ~.
	\end{equation*}
\end{lemma}

\begin{proof}[Proof Sketch of Lemma~\ref{lemma:l(xt,pt)-l(xt,p)}]
    We present a proof sketch of a weaker claim that the inequality holds in expectation, which captures the bottleneck of the analysis.
    Further showing the stronger claim in the lemma takes a (slightly nonstandard) concentration inequality for martingales.
    See Appendix~\ref{sec:miss-proof-utility} for a complete proof of the lemma.

    Readers familiar with this kind of analysis will find it standard and may directly jump to the end to verify that the contribution from the variance term can be bounded given the noise scales chosen in the algorithm.
    Fix any step $t$, we have the followings:
    \begin{align*}
        \E \big[ \eta^t \big( L(x^t, p^t) - L(x^t, p) \big) \big]
        & 
        = \E \big[ \langle \eta^t \nabla D(p^t), p^t - p \rangle \big]\\
        &
        = \E \big[ \langle \delta^t, p^t - p \rangle \big]\\
        &
        = \E \left[ \big\langle \ln \big( \tfrac{\hat{p}^{t+1}}{p^t} \big), p - p^t \big\rangle \right] \\
        &
        = \E \big[ \kl\big(p\|p^t\big) - \kl\big(p\|\hat{p}^{t+1}\big) + \kl\big(p^t\|\hat{p}^{t+1}\big) \big] \\[1ex]
        &
        \le \E \big[ \kl\big(p\|p^t\big) - \kl\big(p\|p^{t+1}\big) + \kl\big(p^t\|\hat{p}^{t+1}\big) - \kl\big(p^{t+1}\|\hat{p}^{t+1}\big) \big] \\[1ex]
        &
        \le \E \big[ \kl\big(p\|p^t\big) - \kl\big(p\|p^{t+1}\big) + \kl\big(p^t\|\hat{p}^{t+1}\big) \big]
        ~.
    \end{align*}

    We abuse notation and let $\ln\big(\frac{\hat{p}^{t+1}}{p^t}\big)$ denote a vector whose $j$-th coordinate is $\ln\big(\frac{\hat{p}^{t+1}_j}{p^t_j}\big)$ for any $j \in [m+1]$, in the 4th line of the above equation.
    The last two inequalities follow by the generalized Pythagorean theorem and the non-negativity of divergences, respectively.

    Summing over $t \in [T]$, the first two terms form a telescopic sum; 
    it is bounded by the first term on the right-hand-side of the inequality stated in the lemma.
    For the last term, we have:
    \begin{align*}
        \E \big[ \kl\big(p^t\|\hat{p}^{t+1}\big) \big]
        &
        \textstyle
        = \E \big[ \sum_{j=1}^{m+1} \big( p^t_j \ln\big(\tfrac{p^t_j}{\hat{p}^{t+1}_j}\big) - p^t_j + \hat{p}^{t+1}_j \big) \big] \\
        &
        \textstyle
        = \E \big[ \sum_{j=1}^{m+1} p^t_j \big( \delta^t_j - 1 + \exp\big(-\delta^t_j\big) \big) \big] \\
        &
        \textstyle
        \le \E \big[ \sum_{j=1}^{m+1} p^t_j \big( \delta^t_j \big)^2 \big]
        = \E \left[ \sum_{j=1}^{m+1} p^t_j \big( \E \big[ \delta^t_j \big]^2 + \Var \big[ \delta^t_j \big] \big) \right]
        ~.
    \end{align*}
    
    Note that the step sizes ensure $\E \big[ \delta^t_j \big] \le \alpha$.
    The first part on the right-hand-side sums to:
    \begin{align*}
        \textstyle
        \E \left[ \sum_{j=1}^{m+1} \sum_{t=1}^T p^t_j \cdot \E \big[ \delta^t_j \big]^2 \right]
        &    
        \textstyle
        \le \alpha \E \left[ \sum_{j=1}^{m+1} p^t_j \cdot \sum_{t=1}^T \big| \E \big[ \delta^t_j \big] \big| \right] \\
        &
        \textstyle
        \le \alpha \E \left[ \sum_{j=1}^{m+1} p^t_j \cdot \left( 2b \sum_{t=1}^T \eta^t - \sum_{t=1}^T \eta^t \E \big[ \delta^t_j \big] \right) \right]\\
        &
        \textstyle
        = \etasum \cdot O(\alpha n) - \alpha \E \left[ \sum_{j=1}^{m+1} p^t_j \sum_{t=1}^T \eta^t \E \big[ \delta^t_j \big] \right] \\
        &
        \textstyle
        \le \etasum \cdot O(\alpha n) - \alpha \E \left[ \sum_{j=1}^{m+1} p^t_j \sum_{t=1}^T \eta^t \nabla_j D(p^t) \right] \\
        &
        \textstyle
        = \etasum \cdot O(\alpha n) - \alpha \E \left[ \sum_{j=1}^{m+1} \eta^t \big( L(x^t, p^t) - \sum_{i=1}^n v_i x^t_i \big) \right] \\
        &
        \textstyle
        \le \etasum \cdot O(\alpha n) - \alpha \E \left[ \sum_{j=1}^{m+1} \eta^t L(x^t, p^t) \right]
        \le \etasum \cdot O(\alpha n) 
        ~.
    \end{align*}

    Finally, the contribution from the variance part sums to:
    \begin{align*}
        \textstyle
        \E \left[ \sum_{j=1}^{m+1} \sum_{t=1}^T p^t_j \cdot \Var \big[ \delta^t_j \big] \right]
        &
        \textstyle
        \le \E \left[ \sum_{j=1}^{m+1} \sum_{t=1}^T p^t_j \cdot 4 \big(\sigma^t\big)^2 \right]
        && 
        \textrm{(Lemma~\ref{lem:noise-distributions-privacy})} \\
        &
        \textstyle
        \le \tilde{O} \big( \sum_{t=1}^T \frac{\pmax m \etasum \eta^t}{\epsilon^2} \big) 
        &&
        \textrm{(Definition of $\sigma^t$'s)} \\
        &
        \textstyle
        \le \tilde{O} \big( \frac{\pmax m \etasum^2}{\epsilon} \big) \\
        &
        \textstyle
        = \etasum \cdot \tilde{O} \big( \frac{n}{\alpha \epsilon^2 b^2} \big) \le \etasum \cdot O(\alpha n)
        ~.
        &&
        \textrm{(Assumption on $b$)}
    \end{align*}

    Putting them together proves the lemma.
\end{proof}

By the choice of $\vec{p}^1$, $\pmax$, $\etasum$, the fact that $L(x^t, p^t) \ge \opt$ since $x^t$'s are best responses, and the definition of $\bar{x}$, we further have it in a simpler form as a corollary.

\begin{lemma}
    \label{lem:no-regret}
	For any $p$ such that $\|p\|_1=p_{max}$, with high probability, we have:
	\begin{equation*}
        \opt - L \big( \bar{x}, p \big) 
        \leq 
        O \big( \alpha n \big) 
        ~.
	\end{equation*}
\end{lemma}

\subsection{Approximate Optimality}

We now argue that the algorithm gets an objective that is optimal up to an $O(\alpha n)$ additive factor, with the understanding that further improving it to an $\alpha n$ factor does not affect any of the asymptotic bound.
To do so, simply let $p$ be such that the first $m$ coordinates are all equal to $0$, and the last dummy coordinate equals $\pmax$.
Then, we have that $L(\bar{x}, p) = \sum_{i=1}^n v_i \bar{x}_i = \alg$.
The claim then follows from Lemma~\ref{lem:no-regret}.

\subsection{Feasibility}
\label{subsec:feasibility}

Next, we argue that the algorithm provides an allocation $\bar{x}$ that is approximately feasible in the sense that the total demand for each resource is at most $\big(1 + O(\alpha) \big) b$.
To convert it into an exactly feasible solution as stated in Theorem~\ref{thm:scalable-optimal-jdp-packing}, we can simply scale the allocation down by a $1 + O(\alpha)$ factor, at the cost of further reducing the objective by at most $O(\alpha n)$.

Suppose resource $j^*$ has the largest demand.
Let $s = \sum_{i=1}^n a_{ij^*} \bar{x}_i - b$ the gap between this demand and the supply $b$.
Let $p$ be such that all but the $j^*$-th coordinate are equal to $0$, and the $j^*$-th coordinate is equal to $\pmax$.
Then, we have:
\[
    L \big( \bar{x}, p \big) = \alg - \pmax s \le \left(1 + \tfrac{s}{b} \right) \opt - \pmax s
    ~.
\]

By Lemma~\ref{lem:no-regret}, we get that:
$
    \left( \pmax - \tfrac{\opt}{b} \right) s \le O \big( \alpha n \big) 
    ~.
$

Putting together with our choice of $\pmax = \frac{2n}{b}$, and that $\opt \le n$, we get $s \le O \big( \alpha b \big)$.

\subsection{Time Complexity}
\label{sec:time}

We next show that the number of iterations of the for loop is upper bounded by a polynomial of the parameters other than $n$.
More precisely, we show that $T \le O \big( \frac{m \ln(m+1)}{\alpha^2} \big)$.

To do so, we classify the iterations into $m+1$ types, according to how the step size is chosen.
If the step size is $\eta^t = \frac{\alpha}{b}$, we call it type $0$.
Otherwise, if the step size is $\eta^t = \frac{\alpha}{\nabla_j D(p^t)}$ for some $j \in [m]$, we call it type $j$.

First, the number of type $0$ iterations is at most $\frac{\ln(m+1)}{\alpha^2}$, before the sum of the step sizes exceeds $\etasum = \frac{\ln(m+1)}{\alpha b}$.

Next, consider the iterations of type $j$, for each resource $j \in [m]$.
In each iteration $t$ of type $j$, we have that $\eta^t \nabla_j D(p^t) = \alpha$.
We argue that there cannot be too many such iterations, because $\sum_{t=1}^T \eta^t \nabla_j D(p^t)$ is upper bounded as a result of the approximate feasibility of $\bar{x}$, which follows from Lemma~\ref{lem:no-regret} as argued in the Section~\ref{subsec:feasibility}.

On the one hand, we have:
\[
    \textstyle
\sum_{t=1}^T \eta^t \nabla_j D(p^t) = \sum_{t=1}^T \eta^t \big( \sum_{i=1}^n a_{ij} \bar{x}_i - b \big) \le \etasum \cdot O(\alpha b) = \tilde{O} \big( 1 \big)
    ~.
\]

On the other hand, suppose there are $T_j$ such iterations, we have:
\begin{align*}
    \textstyle
    \sum_{t=1}^T \eta^t \nabla_j D(p^t)
    & 
    \textstyle
    = \sum_{t \in T_j} \alpha + \sum_{t \notin T_j} \eta^t \nabla_j D(p^t) \\
   &
    \textstyle
    \ge \sum_{t \in T_j} \alpha + \sum_{t \notin T_j} \eta^t (-b) 
    \textstyle
    \ge \alpha |T_j| - 2 b \etasum = \alpha |T_j| - \frac{2 \ln(m+1)}{\alpha}
    ~.
\end{align*}

Putting together gives $|T_j| \le \frac{3 \ln(m+1)}{\alpha^2}$.

\section{Privacy: Proof Sketch}
\label{sec:privacy}

This section sketches the proof that Algorithm~\ref{alg:mw-optimized-step-sizes} is $(\epsilon, \delta)$-jointly differentially private.
See Appendix~\ref{app:privacy-full} for a complete argument.

First, consider any fix step $t$.
Suppose $\mathcal{D}$ and $\widetilde{\mathcal{D}}$ are two neighboring datasets.
Further suppose the random realizations of the first $t-1$ iterations are such that $x^1, x^2, \dots, x^{t-1}$ and $p^1, p^2, \dots, p^{t-1}$ are identical on $\mathcal{D}$ and $\widetilde{\mathcal{D}}$.
Let us consider the privacy of step $t$ alone.

In fact, let us focus on a resource $j \in [m]$.
How does the distribution from which $\delta^t_j$ is drawn differ on the two datasets?
Suppose the step sizes are $\eta$ and $\tilde{\eta}$ respectively.
Then, by the choice of step sizes in the algorithm and that $\nabla_j D(p^t)$ differs by at most $1$, we have:
\[
    \big| \tfrac{1}{\eta} - \tfrac{1}{\tilde{\eta}} \big| \le \tfrac{1}{\alpha}
    ~.
\]

Suppose the standard deviations are $\sigma$ and $\tilde{\sigma}$, and the means are $\mu$ and $\tilde{\mu}$, respectively on the two datasets.
They satisfy that:
\begin{align*}
    \big| \tfrac{1}{\sigma^2} - \tfrac{1}{\tilde{\sigma}^2} \big| 
    & 
    = \tfrac{\epsilon^2}{m \etasum \ln(T/\delta)} \big| \tfrac{1}{\eta} - \tfrac{1}{\tilde{\eta}} \big|
    \le \tfrac{\epsilon^2}{\alpha m \etasum \ln(Tm/\delta)} = \tfrac{\eta}{\alpha \sigma^2}
    ~, \\
    \big| \mu - \tilde{\mu} \big| 
    &
    \le \eta \big|  \tfrac{\mu}{\eta} - \tfrac{\tilde{\mu}}{\tilde{\eta}} \big| + \tilde{\mu} \eta \big| \tfrac{1}{\eta} - \tfrac{1}{\tilde{\eta}} \big| \le \eta + \alpha \eta \alpha^{-1} = 2 \eta 
    ~,
\end{align*}
where $\big|  \tfrac{\mu}{\eta} - \tfrac{\tilde{\mu}}{\tilde{\eta}} \big| \le 1$ because it measures the difference in $\nabla_j D(p^t)$.

Also it's easy to verify that $\eta^t\leq \alpha$. Therefore, by Lemma~\ref{lem:noise-distributions-privacy}, we have that the update from price $p^t_j$ to $\tilde{p}^{t+1}_j$ posted in round $t$ for resource $j$ is $(\epsilon^t, \frac{\delta}{Tm})$-differentially private for:
\[
    \epsilon^t = O \big( \tfrac{\eta^t}{\sigma^t} \big) = \tfrac{\epsilon \sqrt{\eta^t} \ln(Tm/\delta)}{\sqrt{m \etasum}}
    ~.
\]

Then, intuitively by the composition theorem (see, e.g., Dwork and Roth \cite{DworkR/2014/FTTCS}), the price sequence is $(\epsilon, \delta)$-differentially private because: 
\[
    \sum_{t=1}^T \sum_{j=1}^m \big( \epsilon^t \big)^2 = O \big( \tfrac{\epsilon^2 \eta^t \ln^2(Tm/\delta)}{m \etasum} \big) = O\big(\tfrac{\epsilon^2}{\ln(2/\delta)}\big)
    ~.
\]

However, note that the standard statement of the composition theorem does not directly apply here since the privacy budget $\epsilon^t$ in each step is chosen adaptively. 
Fortunately, the underlying argument still goes through. 
Again, see Appendix~\ref{app:privacy-full} for details.

Finally, the privacy guarantee of Algorithm~\ref{alg:mw-optimized-step-sizes} follows by the billboard lemma (Lemma~\ref{lem:billboard}).

\bibliographystyle{plain}
\bibliography{jdp}

\begin{thebibliography}{10}

\bibitem{Afriat/1971/SIAP}
SN~Afriat.
\newblock Theory of maxima and the method of lagrange.
\newblock {\em SIAM Journal on Applied Mathematics}, 20(3):343--357, 1971.

\bibitem{AgrawalD/2014/SODA}
Shipra Agrawal and Nikhil~R Devanur.
\newblock Fast algorithms for online stochastic convex programming.
\newblock In {\em SODA}, pages 1405--1424. SIAM, 2014.

\bibitem{BassilyS/2015/STOC}
Raef Bassily and Adam Smith.
\newblock Local, private, efficient protocols for succinct histograms.
\newblock In {\em STOC}, pages 127--135. ACM, 2015.

\bibitem{ChanSS/2010/ICALP}
TH~Hubert Chan, Elaine Shi, and Dawn Song.
\newblock Private and continual release of statistics.
\newblock In {\em ICALP}, pages 405--417. Springer, 2010.

\bibitem{CummingsIL/2015/COLT}
Rachel Cummings, Stratis Ioannidis, and Katrina Ligett.
\newblock Truthful linear regression.
\newblock In {\em COLT}, pages 448--483, 2015.

\bibitem{CummingsKRW/2015/WINE}
Rachel Cummings, Michael Kearns, Aaron Roth, and Zhiwei~Steven Wu.
\newblock Privacy and truthful equilibrium selection for aggregative games.
\newblock In {\em WINE}, pages 286--299. Springer, 2015.

\bibitem{CummingsLRRW/2016/ITCS}
Rachel Cummings, Katrina Ligett, Jaikumar Radhakrishnan, Aaron Roth, and
  Zhiwei~Steven Wu.
\newblock Coordination complexity: Small information coordinating large
  populations.
\newblock In {\em ITCS}, pages 281--290. ACM, 2016.

\bibitem{CummingsPV/2016/EC}
Rachel Cummings, David~M Pennock, and Jennifer~Wortman Vaughan.
\newblock The possibilities and limitations of private prediction markets.
\newblock In {\em EC}, pages 143--160. ACM, 2016.

\bibitem{DuchiJW/2013/FOCS}
John~C Duchi, Michael~I Jordan, and Martin~J Wainwright.
\newblock Local privacy and statistical minimax rates.
\newblock In {\em FOCS}, pages 429--438. IEEE, 2013.

\bibitem{DworkMNS/2006/TCC}
Cynthia Dwork, Frank McSherry, Kobbi Nissim, and Adam Smith.
\newblock Calibrating noise to sensitivity in private data analysis.
\newblock In {\em TCC}, pages 265--284. Springer, 2006.

\bibitem{DworkNPR/2010/STOC}
Cynthia Dwork, Moni Naor, Toniann Pitassi, and Guy~N Rothblum.
\newblock Differential privacy under continual observation.
\newblock In {\em STOC}, pages 715--724. ACM, 2010.

\bibitem{DworkR/2014/FTTCS}
Cynthia Dwork and Aaron Roth.
\newblock The algorithmic foundations of differential privacy.
\newblock {\em Foundations and Trends{\textregistered} in Theoretical Computer
  Science}, 9(3--4):211--407, 2014.

\bibitem{GhoshLRS/2014/EC}
Arpita Ghosh, Katrina Ligett, Aaron Roth, and Grant Schoenebeck.
\newblock Buying private data without verification.
\newblock In {\em EC}, pages 931--948. ACM, 2014.

\bibitem{HsuHRRW/2016/SICOMP}
Justin Hsu, Zhiyi Huang, Aaron Roth, Tim Roughgarden, and Zhiwei~Steven Wu.
\newblock Private matchings and allocations.
\newblock {\em SIAM Journal on Computing}, 45(6):1953--1984, 2016.

\bibitem{HsuHRW/2016/SODA}
Justin Hsu, Zhiyi Huang, Aaron Roth, and Zhiwei~Steven Wu.
\newblock Jointly private convex programming.
\newblock In {\em SODA}, pages 580--599. SIAM, 2016.

\bibitem{HsuRRU/2014/ICALP}
Justin Hsu, Aaron Roth, Tim Roughgarden, and Jonathan Ullman.
\newblock Privately solving linear programs.
\newblock In {\em ICALP}, pages 612--624. Springer, 2014.

\bibitem{HuangZ/2018/SODA}
Zhiyi Huang and Xue Zhu.
\newblock Near optimal jointly private packing algorithms via dual
  multiplicative weight update.
\newblock In {\em SODA}, pages 343--357. SIAM, 2018.

\bibitem{KearnsPRU/2014/ITCS}
Michael Kearns, Mallesh Pai, Aaron Roth, and Jonathan Ullman.
\newblock Mechanism design in large games: incentives and privacy.
\newblock In {\em ITCS}, pages 403--410. ACM, 2014.

\bibitem{KoufogiannakisY/2014/Algorithmica}
Christos Koufogiannakis and Neal~E Young.
\newblock A nearly linear-time ptas for explicit fractional packing and
  covering linear programs.
\newblock {\em Algorithmica}, 70(4):648--674, 2014.

\bibitem{LiuHW/2018/NeurIPS}
Jinyan Liu, Zhiyi Huang, and Xiangning Wang.
\newblock Learning optimal reserve price against non-myopic bidders.
\newblock In {\em NeurIPS}, pages 2042--2052, 2018.

\bibitem{McSherryT/2007/FOCS}
Frank McSherry and Kunal Talwar.
\newblock Mechanism design via differential privacy.
\newblock In {\em FOCS}, pages 94--103. IEEE Computer Society, 2007.

\bibitem{RogersRUW/2015/EC}
Ryan Rogers, Aaron Roth, Jonathan Ullman, and Zhiwei~Steven Wu.
\newblock Inducing approximately optimal flow using truthful mediators.
\newblock In {\em EC}, pages 471--488. ACM, 2015.

\bibitem{RogersR/2014/EC}
Ryan~M Rogers and Aaron Roth.
\newblock Asymptotically truthful equilibrium selection in large congestion
  games.
\newblock In {\em EC}, pages 771--782. ACM, 2014.

\bibitem{rogers2016privacy}
Ryan~M Rogers, Aaron Roth, Jonathan Ullman, and Salil Vadhan.
\newblock Privacy odometers and filters: Pay-as-you-go composition.
\newblock In {\em Advances in Neural Information Processing Systems}, pages
  1921--1929, 2016.

\bibitem{ShariffS/2018/NeurIPS}
Roshan Shariff and Or~Sheffet.
\newblock Differentially private contextual linear bandits.
\newblock In {\em NeurIPS}, pages 4301--4311, 2018.

\bibitem{Teng/2016/FTTCS}
Shang-Hua Teng.
\newblock Scalable algorithms for data and network analysis.
\newblock {\em Foundations and Trends{\textregistered} in Theoretical Computer
  Science}, 12(1--2):1--274, 2016.

\end{thebibliography}
\appendix
\section{Extended Azuma-Hoeffding Inequality}
\label{app:concentration}
In this section we prove an extension of the Azuma-Hoeffding inequality for martingales. 

\begin{theorem}
	\label{thm:martingale-concentration}
	Suppose $X^1, X^2, \dots, X^T$ is a martingale w.r.t.\ $r^1, r^2, \dots, r^T$ such that $|X^t - X^{t-1}| \le 1$ for any $t \in [T]$ (with $X^0 = 0$), and that:
	\[
	\sum_{t = 1}^T \Var \big[ X^t ~\big|~ r^{1:t-1} \big] \le \sigma^2
	\]
	Then, we have that:
	\begin{equation}
	\label{eqn:martingale-bounded-total-variance}
	\Pr \big[ X^T \ge \alpha \big] \le \exp \big( - \frac{\alpha^2}{4 \sigma^2} \big) 
	~.
	\end{equation}
\end{theorem}

\begin{proof}
	For some parameter $0 < \lambda < 1$ to be determined in the analysis, we have the following sequence of inequalities:
	\begin{align*}
	\Pr \big[ X^T \ge \alpha \big] 
	& 
	= \Pr \big[ \exp \big( \lambda (X^T - \alpha) \big) \ge 1 \big] \\
	& \le \E \big[ \exp \big( \lambda (X^T - \alpha) \big) \big] \\
	& = \exp \big( - \lambda \alpha \big) \cdot \E \big[ \exp \big( \lambda X^T \big) \big]
	~.
	\end{align*}
	
	The main technical ingredient is an upper bound on the second part of the RHS, which we shall prove separately as Lemma~\ref{lem:martingale-induction}.
	By this lemma, we further get that:
	\begin{align*}
	\Pr \big[ X^T \ge \alpha \big] 
	& \le \exp \big( - \lambda \alpha \big) \cdot \exp \big( \lambda^2 \sigma^2 \big) \\
	& = \exp \big( - \lambda \alpha + \lambda^2 \sigma^2 \big) 
	~.
	\end{align*}
	
	Let $\lambda = \frac{\alpha}{2 \sigma^2}$ to minimize the exponent, we get that:
	\[
	\Pr \big[ X^T \ge \alpha \big] \le \exp \bigg( - \frac{\alpha^2}{4\sigma^2} \bigg)
	~.
	\]
\end{proof}

\begin{lemma}
	\label{lem:martingale-induction}
	For any martingale as defined in Theorem~\ref{thm:martingale-concentration}, and any $0 < \lambda < 1$, we have:
	\[
	\E \big[ \exp \big( \lambda X^T \big) \big] \le \exp \big( \lambda^2 \sigma^2 \big) 
	~.
	\]
\end{lemma}

\begin{proof}
	We will prove by induction that, for any $0 \le k \le T$, for any realization of the random bits $r^{1:T-k}$, the following holds over the randomness of the remaining random bits $r^{T-k+1:T}$:
	\[
	\E_{r^{T-k+1:T}} \big[ \exp \big( \lambda ( X^T - X^{T-k} ) \big) ~\big|~ r^{1:T-k} \big] \le \exp \bigg( \lambda^2 \bigg( \sigma^2 - \sum_{t=1}^{T-k} \Var \big[ X^t ~\big|~ r^{1:t-1} \big] \bigg) \bigg)
	~.
	\]
	
	Then, the lemma follows by the case when $k = T$.
	
	The base case of $k = 0$ is equivalent to Eqn.~\eqref{eqn:martingale-bounded-total-variance} in the Theorem statement.
	
	Next, suppose the theorem holds for some $0 \le k < T$.
	The case of $k+1$ then follows by the following sequence of inequalities:
	\begin{align*}
	\E_{r^{T-k:T}} & \big[ \exp \big( \lambda ( X^T - X^{T-k-1} ) \big) ~\big|~ r^{1:T-k-1} \big] \\[2ex]
	&
	= \E_{r^{T-k}} \bigg[ \exp \bigg( \lambda ( X^{T-k} - X^{T-k-1} ) \bigg) \cdot \E_{r^{T-k+1:T}} \big[ \exp \big( \lambda (X^T - X^{T-k}) \big) ~\big|~ r^{1:T-k} \big] ~\bigg|~ r^{1:T-k-1} \bigg] \\
	&
	\le \E_{r^{T-k}} \bigg[ \exp \bigg( \lambda ( X^{T-k} - X^{T-k-1} ) \bigg) \cdot \exp \bigg( \lambda^2 \bigg( \sigma^2 - \sum_{t=1}^{T-k} \Var \big[ X^t ~\big|~ r^{1:t-1} \big] \bigg) \bigg) ~\bigg|~ r^{1:T-k-1} \bigg] \\
	& 
	\le \E_{r^{T-k}} \bigg[ \exp \bigg( \lambda ( X^{T-k} - X^{T-k-1} ) \bigg) ~\bigg|~ r^{1:T-k-1} \bigg] \cdot \exp \bigg( \lambda^2 \bigg( \sigma^2 - \sum_{t=1}^{T-k} \Var \big[ X^t ~\big|~ r^{1:t-1} \big] \bigg) \bigg) \\
	&
	\le \E_{r^{T-k}} \bigg[ \bigg( 1 + \lambda (X^{T-k} - X^{T-k-1}) + \lambda^2 (X^{T-k} - X^{T-k-1})^2 \bigg) ~\bigg|~ r^{1:T-k-1} \bigg] \\
	&
	\qquad\qquad
	\cdot \exp \bigg( \lambda^2 \bigg( \sigma^2 - \sum_{t=1}^{T-k} \Var \big[ X^t ~\big|~ r^{1:t-1} \big] \bigg) \bigg) \\
	&
	= \bigg( 1 + \lambda^2 \Var \big[ X^{T-k} ~\big|~ r^{1:T-k-1} \big] \bigg) \cdot \exp \bigg( \lambda^2 \bigg( \sigma^2 - \sum_{t=1}^{T-k} \Var \big[ X^t ~\big|~ r^{1:t-1} \big] \bigg) \bigg) \\
	&
	\le \exp \bigg( \lambda^2 \Var \big[ X^{T-k} ~\big|~ r^{1:T-k-1} \big] \bigg) \cdot \exp \bigg( \lambda^2 \bigg( \sigma^2 - \sum_{t=1}^{T-k} \Var \big[ X^t ~\big|~ r^{1:t-1} \big] \bigg) \bigg) \\
	& 
	= \exp \bigg( \sigma^2 - \sum_{t=1}^{T-k-1} \Var \big[ X^t ~\big|~ r^{1:t-1} \big] \bigg)
	~.
	\end{align*}
\end{proof}
\section{Missing Proofs in the Truncated Laplacian Distribution}
\label{sec:miss-proof-truncated}
In this section, we will prove lemma \ref{lem:noise-distributions-statistics} and lemma \ref{lem:noise-distributions-privacy}.  
\subsection{Proof of lemma \ref{lem:noise-distributions-statistics}: Mean and Variance}
For the expectation, it's easy to verity that the expectation of $\dist(\mu,\sigma)$ is equal to that of $Lap(\mu,\sigma)$ since it's truncated symmetrically.

Let the standard deviation of $\dist(\mu, \sigma)$ be $\hat{\sigma}$, then we have
\begin{align}
\label{eq:variance}
\begin{split}
\hat{\sigma}^2 & = \textstyle \frac{1}{1-\exp(\frac{-1+\alpha}{\sigma})}\int_{\mu-1+\alpha}^{\mu+1-\alpha} (x-\mu)^2 \frac{1}{2\sigma}\exp(-\frac{|x-\mu|}{\sigma}) \, dx\\
&\textstyle =\frac{1}{(1-\exp(\frac{-1+\alpha}{\sigma}))} \int_{\mu-1+\alpha}^{\mu} (x-\mu)^2 \frac{1}{\sigma}\exp(\frac{x-\mu}{\sigma}) \, dx ~.\\
\end{split}
\end{align}
Note that 
\begin{align*}
&\textstyle  \int_{\mu-1+\alpha}^{\mu} (x-\mu)^2 \frac{1}{\sigma}\exp(\frac{x-\mu}{\sigma}) \, dx\\
&\textstyle =(x-\mu)^2\exp(\frac{x-\mu}{\sigma})|_{\mu-1+\alpha}^{\mu}-\int_{\mu-1+\alpha}^{\mu}2(x-\mu)\exp(\frac{x-\mu}{\sigma})\,dx\\
&\textstyle =-(-1+\alpha)^2\exp(\frac{-1+\alpha}{\sigma})-\left(2(x-\mu)\sigma\exp(\frac{x-\mu}{\sigma})|_{\mu-1+\alpha}^{\mu}-\int_{\mu-1+\alpha}^{\mu}2\sigma\exp(\frac{x-\mu}{\sigma})\, dx\right)\\
&\textstyle =-(-1+\alpha)^2\exp(\frac{-1+\alpha}{\sigma})+2(-1+\alpha)\sigma\exp(\frac{-1+\alpha}{\sigma})+2\sigma^2\exp(\frac{x-\mu}{\sigma})|_{\mu-1+\alpha}^{\mu}\\
&\textstyle =-(-1+\alpha)^2\exp(\frac{-1+\alpha}{\sigma})+2(-1+\alpha)\sigma\exp(\frac{-1+\alpha}{\sigma})+2\sigma^2-2\sigma^2\exp(\frac{-1+\alpha}{\sigma})\\
&\textstyle =2\sigma^2+\exp(\frac{-1+\alpha}{\sigma})\left(2(-1+\alpha)\sigma-2\sigma^2-(-1+\alpha)^2\right)
\end{align*}
Plug in it to (\ref{eq:variance}), we have, 
\begin{align*}
\textstyle
\hat{\sigma}^2  =\frac{1}{(1-\exp(\frac{-1+\alpha}{\sigma}))}\left(2\sigma^2+\exp(\frac{-1+\alpha}{\sigma})\left(2(-1+\alpha)\sigma-2\sigma^2-(-1+\alpha)^2\right)\right)
\leq 2\sigma^2
\end{align*}
\subsection{Proof of Lemma \ref{lem:noise-distributions-privacy}: Privacy Property}
Recall the statement of lemma \ref{lem:noise-distributions-privacy}:

\begin{quote}
	\em
	For $\dist(\mu_1,\sigma_1), \dist(\mu_2,\sigma_2)$, suppose $-\alpha \le \mu_1, \mu_2 \le \alpha$, and $0 < \sigma_1, \sigma_2 \le \alpha$ satisfy that for some   $0<\eta \leq \alpha$, $\sigma = \max \{\sigma_1, \sigma_2\}$, and $\delta > 0$: (1) $\big| \frac{1}{\sigma_1^2} - \frac{1}{\sigma_2^2} \big| \le \frac{\eta}{\alpha \sigma^2}$; (2) $\big| \mu_1 - \mu_2 \big| \le \eta$; and $\delta \le \frac{\eta \ln(2/\delta)}{\sigma}$.
	Then, for any $S \subseteq \R$, we have:
	\[
	\Pr_{z \sim \dist(\mu_1, \sigma_1)} \big[ z \in S \big] \le \exp\left(\frac{4 \eta \ln(2/\delta)}{\sigma}\right) \cdot \Pr_{z \sim \dist(\mu_2, \sigma_2)} \big[ z \in S \big] + \delta
	~.
	\]
\end{quote}
\begin{proof}	
Firstly, the probability density function of $\dist(\mu, \sigma)$ is
$f(x)=\frac{g(x)}{\Phi(\mu+1-\alpha)-\Phi(\mu-1+\alpha)}$, where $g(x)$ is the pdf of the original Laplacian distribution, $g(x)=\frac{1}{2\sigma}\exp(-\frac{|x-\mu|}{\sigma})$, $\Phi(y)$ is the cumulative density function of the original Laplacian distribution, $ \Phi(\mu+1-\alpha)-\Phi(\mu-1+\alpha)=1-\exp(\frac{-1+\alpha}{\sigma})$. So the pdf of $\dist(\mu, \sigma)$ is
$$f(x,\mu,\sigma)=\frac{\frac{1}{2\sigma}\exp(-\frac{|x-\mu|}{\sigma})}{1-\exp(\frac{-1+\alpha}{\sigma})} ~.$$
To prove lemma \ref{lem:noise-distributions-privacy}, it suffices to show that the cumulative density functions  $F (S,\mu_1,\sigma_1)$ and   $F(S,\mu_2,\sigma_2)$  are similar for any  $S\subseteq \R$, formally,

\[F(S,\mu_1, \sigma_1)\leq  \exp\left(\frac{4 \eta \ln(2/\delta)}{\sigma}\right) \cdot F(S,\mu_2,\sigma_2)+ \delta ~.\]

In order to show that $F(S,\mu_1,\sigma_1)$ is similar to $F(S,\mu_2,\sigma_2)$, we firstly show that $F(S,\mu_1,\sigma_1)$ is similar to $F(S,\mu_2,\sigma_1)$. It means that if we truncate two Laplacian distributions with same variance and different expectations according to our parameter choosing, then the cumulative density functions are similar, as shown in lemma \ref{lemma:same variance}. Secondly, we show that  $F(S,\mu_2,\sigma_1)$ is similar to $F(S,\mu_2,\sigma_2)$. It means that if we truncate two Laplacian distributions with same expectation and different variances according to our parameter choosing, then the cumulative density functions are similar, as shown in lemma \ref{lemma:same expectation}. The following equation will be more straightforward to understand that:
\[
\frac{F(S,\mu_1,\sigma_1)}{F(S,\mu_2,\sigma_2)}=\frac{F(S,\mu_1,\sigma_1)}{F(S,\mu_2,\sigma_1)}\cdot \frac{F(S,\mu_2,\sigma_1)}{F(S,\mu_2,\sigma_2)} .\]

\begin{lemma}
	\label{lemma:same variance}
	For the parameters stated above, we have that for any $S\subseteq \R$,
	\[\textstyle 
	F(S,\mu_1,\sigma_1)\leq \exp\left(\frac{\eta \ln(2/\delta)}{\sigma}\right)F(S,\mu_2,\sigma_1) ~.
	\] 
\end{lemma}
\begin{proof}
	We will prove it point-wisely. For any $x\in S$, we have
	\begin{equation*}
	\textstyle
	\frac{f(x,\mu_1,\sigma_1)}{f(x,\mu_2,\sigma_1)}=\exp\left(\frac{|x-\mu_2|}{\sigma_1}-\frac{|x-\mu_1|}{\sigma_1}\right)\leq \exp\left(\frac{|\mu_2-\mu_1|}{\sigma_1}\right)\leq\exp(\frac{\eta}{\sigma_1})\leq \exp\left(\frac{\eta \ln(2/\delta)}{\sigma}\right).
	\end{equation*}
	The second inequality follows by our assumption that $|\mu_1-\mu_2|\leq \eta$ and the last inequality follows by our assumption of $\sigma_1,\sigma_2$ and $\eta$. If it holds at every point, then lemma \ref{lemma:same variance} holds.
\end{proof}

\begin{lemma} 
	\label{lemma:same expectation} 
	For the parameters stated above, 
we have that for any $S\subseteq \R$,
\[
F(S,\mu_2,\sigma_1) \le \exp\left(\frac{4 \eta \ln(2/\delta)}{\sigma}\right) \cdot F(S,\mu_2,\sigma_2)  + \delta
~.
\]

\end{lemma}
\begin{proof}
	Firstly we show that for any $\dist(\mu,\sigma)$ such that $\mu \leq \alpha, \sigma \leq  \alpha$, the probability of point being far away from the center (the expectation) is small, i.e. \ for $S=[\mu-1+\alpha,\mu-\ln(2/\delta)\alpha] \cup [\mu+\ln(2/\delta)\alpha,\mu+1-\alpha]$,
		\begin{align}
		\begin{split}
		\label{eq:far away}
		F(S,\mu,\sigma) &  =\int_{\mu-1+\alpha}^{\mu_1-\ln(2/\delta)\alpha} f(x)  \, dx+\int_{\mu+\ln(2/\delta)\alpha}^{\mu+1-\alpha} f(x) \, dx\\
	&   =\frac{2}{1-\exp(\frac{-1+\alpha}{\sigma})} \int_{\mu-1+\alpha}^{\mu-\ln(2/\delta)\alpha} \frac{1}{2\sigma}\exp(\frac{x-\mu}{\sigma})  \, dx\\
	&  =\frac{1}{1-\exp(\frac{-1+\alpha}{\sigma})}\exp(\frac{x-\mu}{\sigma})\large{|}_{\mu-1+\alpha}^{\mu-\ln(2/\delta)\alpha}\\
	&  =\frac{\exp(\frac{-\ln(2/\delta)\alpha}{\sigma})-\exp(\frac{-1+\alpha}{\sigma})}{1-\exp(\frac{-1+\alpha}{\sigma})}\\
	&  \leq \frac{\exp(\frac{-\ln(2/\delta)\alpha}{\sigma})}{1-\exp(\frac{-1+\alpha}{\sigma})}\leq \delta ~.
	\end{split}
	\end{align}
	We have the last inequality because we assume that $\sigma \leq \alpha$.	
	
	For $S\subseteq[\mu_2-\ln(2/\delta)\alpha,\mu_2+\ln(2/\delta)\alpha]$, we will show it point-wisely. Note that for $x\in S$,
	\begin{align}
	\label{eq:privacy}
 	\frac{f(x,\mu_2,\sigma_1)}{f(x,\mu_2,\sigma_2)}=  \underbrace{ \textstyle \exp\left(\frac{|x-\mu|}{\sigma_2}-\frac{|x-\mu|}{\sigma_1}\right)}_{(\mbox{term }  1)}\cdot \underbrace{ \frac{1-\exp(\frac{-1+\alpha}{\sigma_2})}{1-\exp(\frac{-1+\alpha}{\sigma_1})} \cdot \frac{\sigma_2}{\sigma_1}}_{(\mbox{term } 2)}
	\end{align}

	Then we consider the two terms in (\ref{eq:privacy}) separately.
	
	 For $\mbox{term } 1$, we have
	
	\begin{align}
	\label{eq:privacy1}
	\exp\left(\frac{|x-\mu_2|}{\sigma_2}-\frac{|x-\mu_2|}{\sigma_1}\right)  =\exp\left(|x-\mu_2|\left(\frac{1}{\sigma_2}-\frac{1}{\sigma_1}\right)\right)  \leq \exp(\frac{\eta \ln(2/\delta)}{\alpha})
	\end{align}
	The last inequality holds because $|x-\mu_2|\leq \ln(2/\delta)\alpha$ according to the choosing of $S$, in addition  we assume that$\textstyle |\frac{1}{\sigma_1^2}-\frac{1}{\sigma^2}|\leq \frac{\eta}{\alpha \sigma}$, and note that $\textstyle |\frac{1}{\sigma_1}-\frac{1}{\sigma_2}|\cdot \textstyle  \frac{1}{\sigma}\leq \textstyle  |\frac{1}{\sigma_1^2}-\frac{1}{\sigma_2^2}|\leq \textstyle  \frac{\eta}{\alpha \sigma}$, so we have $\textstyle |\frac{1}{\sigma_1}-\frac{1}{\sigma_2}|\leq \frac{\eta}{\alpha \sigma}$.

	For $\mbox{term } 2$, we use $\Phi(a,b,\mu,\sigma)$ to denote the cumulative density function of the region $[a,b]$ with a Laplacian distribution $(\mu,\sigma)$  :
	\begin{align}
	\begin{split}
	\label{eq:privacy2}
	&\frac{1-\exp(\frac{-1+\alpha}{\sigma_2})}{1-\exp(\frac{-1+\alpha}{\sigma_1})}\cdot \frac{\sigma_2}{\sigma_1} =\frac{\Phi(\mu_2-1+\alpha,\mu_2+1-\alpha,\mu_2,\sigma_2)}{\Phi(\mu_2-1+\alpha,\mu_2+1-\alpha,\mu_2,\sigma_1)}\cdot \frac{\sigma_2}{\sigma_1}\\
	&\leq \frac{\Phi(\mu_2-1+\alpha,\mu_2+1-\alpha,\mu_2,\sigma_2)}{\Phi(\mu_2-\ln(2/\delta)\alpha,\mu_2+\ln(2/\delta)\alpha,\mu_2,\sigma_2)}\cdot \frac{\Phi(\mu_2-\ln(2/\delta)\alpha,\mu_2+\ln(2/\delta)\alpha,\mu_2,\sigma_2)}{\Phi(\mu_2-\ln(2/\delta)\alpha,\mu_2+\ln(2/\delta)\alpha,\mu_2,\sigma_1)}\cdot \frac{\sigma_2}{\sigma_1}\\
	&\leq \frac{\Phi(\mu_2-1+\alpha,\mu_2+1-\alpha,\mu_2,\sigma_2)}{\Phi(\mu_2-\ln(2/\delta)\alpha,\mu_2+\ln(2/\delta)\alpha,\mu_2,\sigma_2)}\cdot \exp \left(\frac{\eta\ln(2/\delta)}{\sigma}\right)\cdot \frac{\sigma_1}{\sigma_2}\cdot \frac{\sigma_2}{\sigma_1}\\
	&=\left(\frac{2\Phi(\mu_2-1+\alpha,\mu_2-\ln(2/\delta)\alpha,\mu_2,\sigma_2)+\Phi(\mu_2-\ln(2/\delta)\alpha,\mu_2+\ln(2/\delta)\alpha,\mu_2,\sigma_2)}{\Phi(\mu_2-\ln(2/\delta)\alpha,\mu_2+\ln(2/\delta)\alpha,\mu_2,\sigma_2)}\right)\exp\left(\frac{\eta\ln(2/\delta)}{\sigma}\right)\\
	&\leq(1+2\delta)\exp\left(\frac{\eta\ln(2/\delta)}{\sigma}\right) \leq \exp\left(\frac{2\eta\ln(2/\delta)}{\sigma}\right) ~.
	\end{split}
	\end{align}
	For the inequality in the third line, in order to compare $\Phi(\mu_2-\ln(2/\delta)\alpha,\mu_2+\ln(2/\delta)\alpha,\mu_2,\sigma_2)$ and $\Phi(\mu_2-\ln(2/\delta)\alpha,\mu_2+\ln(2/\delta)\alpha,\mu_2,\sigma_1)$, we compare it point-wisely. For $x\in[\mu_2-\ln(2/\delta)\alpha,\mu_2+\ln(2/\delta)\alpha]$, we know that 
     \[
      \frac{g(x,\mu_2,\sigma_2)}{g(x,\mu_2,\sigma_1)}
     =\exp\left(\frac{|x-\mu_2|}{\sigma_2}-\frac{|x-\mu_2|}{\sigma_1}\right)\cdot\frac{\sigma_1}{\sigma_2}
     \leq \exp(\frac{\eta\ln(2/\delta)}{\sigma})\cdot \frac{\sigma_1}{\sigma_2} ~,
     \]
     where the last inequality is derived from (\ref{eq:privacy1}).
     
      The first inequality  in the last line of (\ref{eq:privacy2}) holds because we have proved that  $\Phi(\mu-1+\alpha,\mu-\ln(2/\delta)\alpha,\mu,\sigma)\leq \delta/2$, and further more $\frac{\delta}{1-\delta}\leq 2\delta$. The last inequality follows by our assumption that $\delta \le \frac{\eta \ln(2/\delta)}{\sigma}$.

	Finally, we combine  (\ref{eq:privacy}), (\ref{eq:privacy1}), (\ref{eq:privacy2}) together, we get that 
	\begin{align}
	\label{eq:center}
	\mbox{for any }S\in[\mu_2-\ln(2/\delta)\alpha,\mu_2+\ln(2/\delta)\alpha], \quad 
	F(S,\mu_2,\sigma_1)   \leq  \exp(\frac{3\eta\ln(2/\delta)}{\sigma})F(S,\mu_2,\sigma_2) ~.
	\end{align}
	Further more, combining it with (\ref{eq:far away}), we get that for any $S\subseteq \R$,
	\[
	F(S,\mu_2,\sigma_1) \leq \exp(\frac{4\eta\ln(2/\delta)}{\sigma})F(S,\mu_2,\sigma_2) +\delta ~.
	\]
\end{proof}
Combine lemma \ref{lemma:same variance} and lemma \ref{lemma:same expectation} together, we get that
	\[
F(S,\mu_1, \sigma_1)\leq \exp(\frac{4\eta\ln(2/\delta)}{\sigma}) F(S,\mu_2,\sigma_2)+\delta ~,
\]
 which indicates lemma \ref{lem:noise-distributions-privacy}. 
\end{proof}
\section{Proof of Lemma \ref{lemma:l(xt,pt)-l(xt,p)}}
\label{sec:miss-proof-utility}
Recall that under the assumption that the supply is sufficiently large, i.e., $b \ge \tilde{O} \big( \frac{\sqrt{m}}{\alpha \epsilon} \big)$, lemma \ref{lemma:l(xt,pt)-l(xt,p)} holds:
	for any $p$ such that $\|p\|_1=p_{max}$,  with high probability, we have:
	\begin{equation*}
	\textstyle
	\sum_{t=1}^T \eta^t \big( L(x^t,p^t) - L(x^t,p) \big) 
	\leq 
	\kl(p\|p^1) + \etasum \cdot O \big( \alpha n \big) 
	~.
	\end{equation*}
	The analysis goes as follows:
	\begin{align*}
	\eta^t \left( L(\vec{x}^t, \vec{p}^t) - L(\vec{x}^t, \vec{p}) \right) 
	&
	= \langle - \vec{\mu}^t, \vec{p} - \vec{p}^t \rangle \\[1.5ex]
	&
	= \langle - \vec{\delta}^t, \vec{p} - \vec{p}^t \rangle + \langle \vec{\delta}^t - \vec{\mu}^t, \vec{p} - \vec{p}^t \rangle \\
	& 
	= \left\langle \ln \left( \frac{\vec{\hat{p}}^{t+1}}{\vec{p}^t} \right), \vec{p} - \vec{p}^t \right\rangle + \langle \vec{\delta}^t - \vec{\mu}^t, \vec{p} - \vec{p}^t \rangle \\
	& 
	= \kl(\vec{p} \| \vec{p}^t) - \kl(\vec{p} \| \vec{\hat{p}}^{t+1}) + \kl(\vec{p}^t \| \vec{\hat{p}}^{t+1}) + \langle \vec{\delta}^t - \vec{\mu}^t, \vec{p} - \vec{p}^t \rangle \\[1.5ex]
	& 
	\le \kl(\vec{p} \| \vec{p}^t) - \kl(\vec{p} \| \vec{p}^{t+1}) + \kl(\vec{p}^t \| \vec{\hat{p}}^{t+1}) - \kl(\vec{p}^{t+1} \| \vec{\hat{p}}^{t+1}) \\
	& \qquad + \langle \vec{\delta}^t - \vec{\mu}^t, \vec{p} - \vec{p}^t \rangle \\[1.5ex]
	& 
	\le \kl(\vec{p} \| \vec{p}^t) - \kl(\vec{p} \| \vec{p}^{t+1}) + \kl(\vec{p}^t \| \vec{\hat{p}}^{t+1}) + \langle \vec{\delta}^t - \vec{\mu}^t, \vec{p} - \vec{p}^t \rangle
	~.
	\end{align*}
	
	In the 4th line of the above equation, we abuse notation and let $\ln \left( \frac{\vec{\hat{p}}^{t+1}}{\vec{p}^t} \right)$ denote a vector where the $j$-th coordinate equals $\ln \left( \frac{\hat{p}^{t+1}_j}{p^t_j} \right)$. 
	
	The first two terms form a telescopic sum.
	Further, the third term can be upper bounded as follows:
	\begin{align*}
	\kl(\vec{p}^t \| \vec{\hat{p}}^{t+1})
	& 
	= \sum_{j=1}^{m+1} \left( p^t_j \ln \left( \frac{p^t_j}{\hat{p}^{t+1}_j} \right) - p^t_j + \hat{p}^{t+1}_j \right) \\
	& 
	= \sum_{j=1}^{m+1} p^t_j \left( \delta^t_j - 1 + \exp(- \delta^t_j) \right) \\
	&
	\le \sum_{j=1}^{m+1} p^t_j \big( \delta^t_j \big)^2 \\
	& 
	= \sum_{j=1}^{m+1} p^t_j \cdot \big( \eta^t \nabla_j D(\vec{p}^t) \big)^2 + \sum_{j=1}^{m+1} p^t_j \cdot \left( \big( \delta^t_j \big)^2 - \big( \mu^t_j \big)^2 \right)
	~.
	\end{align*}
	
	Summing for $t = 1, 2, \dots, T$, we get that:
	\begin{align}
	\label{eq:L(x,p)}
	\begin{split}
	\sum_{t=1}^T \eta^t \left( L(\vec{x}^t, \vec{p}^t) - L(\vec{x}^t, \vec{p}) \right) 
	&
	\le \kl(\vec{p} \| \vec{p}^1) - D(\vec{p} \| \vec{p}^{T+1}) + \sum_{t=1}^T \langle \vec{\delta}^t - \vec{\mu}^t, \vec{p} - \vec{p}^t \rangle \\
	& 
	\qquad 
	+ \sum_{t=1}^T \sum_{j=1}^{m+1} p^t_j \cdot \big( \eta^t \nabla_j D(\vec{p}^t) \big)^2 + \sum_{t=1}^T \sum_{j=1}^{m+1} p^t_j \cdot \left( \big( \delta^t_j \big)^2 - \big( \mu^t_j \big)^2 \right) \\
	&
	\le \kl(\vec{p} \| \vec{p}^1) + \underbrace{\sum_{t=1}^T \langle \vec{\delta}^t - \vec{\mu}^t, \vec{p} - \vec{p}^t \rangle}_{\text{term 1}} \\
	& 
	\qquad
	+ \underbrace{\sum_{t=1}^T \sum_{j=1}^{m+1} p^t_j \cdot \big( \eta^t \nabla_j D(\vec{p}^t) \big)^2}_{\text{term 2}} + \underbrace{\sum_{t=1}^T \sum_{j=1}^{m+1} p^t_j \cdot \left( \big( \delta^t_j \big)^2 - \big( \mu^t_j \big)^2 \right)}_{\text{term 3}}
	~.
	\end{split}
	\end{align}

	Note that $\kl(\vec{p} \| \vec{p}^1)$ is upper bounded by $\pmax \ln(m+1)$ by the upper bound on the KL divergence between two non-negative vectors with $\ell_1$ norm $\pmax$.
	It remains to upper bound the other three terms, which we will explain in details in the following subsections.
	\subsection{Term 1.}
	For any $t$, let $y^t$ be a random variable defined as follows:
	\[
	y^t = \langle \vec{\delta}^t - \vec{\mu}^t, \vec{p} - \vec{p}^t \rangle/p_{max}
	~.
	\]
	Further, let $Y^t = \sum_{i=1}^t y^i$.
	Note that the distribution from which $\vec{\delta}^t$ is drawn, and the dual vector $\vec{p}^t$ depends only on the realization of the random bits in the first $t-1$ rounds, i.e., $r^{1:t-1}$.
	Moreover, $\vec{\delta}^t$ is subsequently realized using fresh random bits $r^t$.
	Therefore, the sequence $Y^{1:T}$ is a martingale w.r.t.\ the random bits $r^{1:T}$.
	Finally, the stopping criteria of our algorithm ensures that, for any realization of the random bits $r^{1:T}$, the total variance is bounded as follows:
	\begin{align*}
	\sum_{t=1}^T \Var \big[ y^t \big| r^{1:t-1} \big]  
	\le \sum_{t=1}^T  \big(\sigma^t\big)^2 
	\le \sum_{t=1}^T  \frac{m \etasum \eta^t}{\epsilon^2} 
	= \frac{ m \etasum^2}{\epsilon^2}
	~.
	\end{align*}
	It's easy to see that $	\sum_{t=1}^T \Var \big[ Y^t \big| r^{1:t-1} \big] =\sum_{t=1}^T \Var \big[ y^t \big| r^{1:t-1} \big] \le  \frac{ m \etasum^2}{\epsilon^2}$. 
	
	By an extension of the concentration bound for martingales in theorem \ref{thm:martingale-concentration}, we get the following bound with high probability:
	\begin{align}
	\label{eq:l(x,p)-term1}
	Y^T & \le \tilde{O} \bigg( \bigg(\sum_{t=1}^T \Var \big[ y^t \big| r^{1:t-1} \big] \bigg)^{\frac{1}{2}} \bigg) \le \tilde{O} \bigg( \frac{ \sqrt{m} \etasum}{\epsilon} \bigg)  = \tilde{O} \bigg( \frac{ \sqrt{m}}{b \epsilon \alpha} \bigg)
	~.
	\end{align}
	It derives that $$\langle \vec{\delta}^t - \vec{\mu}^t, \vec{p} - \vec{p}^t \rangle=Y^T\cdot p_{max}\leq \tilde{O} \bigg( \frac{p_{max} \sqrt{m}}{b \epsilon \alpha} \bigg) \leq \etasum \cdot O(\alpha n)~,$$
	where the last inequality is derived by plugging in the value of $p_{max},\etasum$ and applying the assumption that $b\geq \tilde{O}\left(\frac{\sqrt{m}}{\alpha \epsilon}\right)$.
	
	\subsection{Term 2}
	
	It is bounded as follows:
	\begin{align}
	\label{eq:l(x,p)-term2}
	\begin{split}
	\sum_{t=1}^T \sum_{j=1}^{m+1} p^t_j \cdot \big( \eta^t \nabla_j D(\vec{p}^t) \big)^2
	&
	\le \alpha \cdot \sum_{t=1}^T \sum_{j=1}^{m+1} \eta^t p^t_j \big| \nabla_j D(\vec{p}^t) \big| \\
	& 
	\le 2 \alpha \cdot \pmax b \sum_{t=1}^T \eta^t - \alpha \cdot \sum_{t=1}^T \sum_{j=1}^{m+1} \eta^t p^t_j \nabla_j D(\vec{p}^t) \\
	&
	\le 2 \alpha \cdot \pmax b \sum_{t=1}^T \eta^t + \alpha \cdot \sum_{t=1}^T \eta^t \cdot n - \alpha \cdot \sum_{t=1}^T \eta^t L( \vec{x}^t, \vec{p}^t) \\
	&
	= \etasum \cdot O \left( 2 \alpha n \right) - \alpha \cdot \sum_{t=1}^T \eta^t L( \vec{x}^t, \vec{p}^t)\\
	&\leq \etasum \cdot O \left( 2 \alpha n \right)
	~.
	\end{split}
	\end{align}
	
	\subsection{Term 3}
	To get a better upper bound of term 3, we further divide it into two parts, its expectation, which is simply the sum of variances in different rounds, and the difference between the realizations of the expected value in each round $t$, which forms a martingale.
	Concretely, we have:
	\begin{align*}
	\sum_{t=1}^T \sum_{j=1}^{m+1} p^t_j \cdot \left( \big( \delta^t_j \big)^2 - \big( \mu^t_j \big)^2 \right)
	& 
	= 
	\sum_{t=1}^T \sum_{j=1}^{m+1} p^t_j \cdot \Var_{r^t} \big[ \delta^t_j ~\big|~ r^{1:t-1} \big] \\
	& 
	\qquad
	+ \sum_{t=1}^T \sum_{j=1}^{m+1} p^t_j \cdot \left( \big( \delta^t_j \big)^2 - \big( \mu^t_j \big)^2 - \Var_{r^t} \big[ \delta^t_j ~\big|~ r^{1:t-1} \big] \right)
	\end{align*}
	
	Consider the 1st term on the RHS.
	By the definition of the algorithm, it equals:
\[
\sum_{t=1}^T \sum_{j=1}^{m+1} p^t \big( \sigma^t \big)^2 
= \sum_{t=1}^T \pmax \big( \sigma^t \big)^2 
= \sum_{t=1}^T \pmax \frac{m \etasum \eta^t}{
	\epsilon^2}
= \frac{\pmax m \etasum^2}{\epsilon^2}
\leq \etasum \cdot O(\alpha n)
~.	\]

	The last inequality is derived from our assumption $b\geq \tilde{O}\left(\frac{\sqrt{m}}{\alpha \epsilon}\right)$.
	
	Next, consider the 2nd term on the RHS.
	All expectations and variances below are taken over the random bits $r^t$ in round $t$ and conditioned on the realization of $r^{1:t-1}$ in the previous rounds.
	For simplicity of notations, we will omit the subscript $r^t$ and the conditions $r^{1:t-1}$ in the following discussions.
	As the analysis are the same for different $t$'s and $j$'s, we will further fix some $t$ and $j$, and omit the superscript $t$ and subscript $j$ below. 
	Finally, it is more convenient to consider $z = \delta - \mu$ as the random variable.
	We have:
	\begin{align*}
	\Var \left[ \delta^2 - \mu^2 - \Var \big[ \delta \big] \right] 
	& 
	= \Var \left[ 2 z \mu + z^2 - \E \big[ z^2 \big] \right] \\
	&
	= \E \left[ \left( 2 z \mu + z^2 - \E \big[ z^2 \big] \right)^2 \right] \\
	&
	= 4 \mu^2 \E \big[ z^2 \big] + 4 \mu \E \big[ z^3 \big] + \E \big[ z^4 \big] - \E \big[ z^2 \big]^2 \\[1ex]
	& 
	\le 4 \mu^2 \E \big[ z^2 \big] + 4 \mu \E \big[ z^3 \big] + \E \big[ z^4 \big] \\[1.5ex]
	&
	\le 4 \mu^2 \E \big[ z^2 \big] + 4 \mu \E \big[ z^2 \big] + \E \big[ z^2 \big] \\[1.5ex]
	& 
	= (2 \mu + 1)^2 \sigma^2 \\[1.5ex]
	& 
	< 2 \sigma^2
	~.
	\end{align*}
	
	The rest of the analysis is identical to that of term 1.
	We get the following upper bound with high probability:
	\begin{align}
	\label{eq:l(x,p)-term3}
	\sum_{t=1}^T \sum_{j=1}^{m+1} p^t_j \cdot \left( \big( \delta^t_j \big)^2 - \big( \mu^t_j \big)^2 - \Var_{r^t} \big[ \delta^t_j ~\big|~ r^{1:t-1} \big] \right)
	\le
	\tilde{O} \bigg( \frac{\pmax \sqrt{m}}{b \epsilon \alpha} \bigg) \leq \etasum O(\alpha n)
	~.
	\end{align}
	
	Combining (\ref{eq:L(x,p)}), (\ref{eq:l(x,p)-term1}), (\ref{eq:l(x,p)-term2}) and (\ref{eq:l(x,p)-term3}) together, lemma \ref{lemma:l(xt,pt)-l(xt,p)} is proved. 

\section{Privacy: Full Proof}
\label{app:privacy-full}

The proof is almost identical to the standard composition theorem of differential privacy.
The only difference is that instead of using the standard Azuma inequality, we need to use our extended version. Also the argument is similar in spirit to that in RRUV \cite{rogers2016privacy}.

We first restate some textbook notations and lemmas in the proof of the standard composition theorem.

Define the max divergence between two random variables $Y$ and $Z$, both of which take value from some domain $\Omega$, to be:
\[
    \maxdiv \big( Y \| Z \big) = \max_{S \subseteq \Omega} \ln \frac{\Pr\big[Y \in S\big]}{\Pr\big[Z \in S\big]}
    ~.
\]

Further, the $\delta$-approximate max divergence is defined to be:
\[
    \maxdiv^\delta \big( Y \| Z \big) =  \max_{S \subseteq \Omega : \Pr [ Y \in S ] > \delta} \ln \frac{\Pr\big[Y \in S\big] - \delta}{\Pr\big[Z \in S\big]}
    ~.
\]

We abuse notation and let $\kl\big(Y \| Z \big)$ denote the KL divergence of the corresponding distributions from which $Y$ and $Z$ are drawn from.

Finally, let $\Delta(Y, Z)$ denote the statistical distance between the two distributions.

\begin{lemma}[Lemma 3.17 of Dwork and Roth \cite{DworkR/2014/FTTCS}]
    For any random variables $Y$ and $Z$ taking values from a common domain $\Omega$, we have:
    \begin{enumerate}
        \item $\maxdiv^\delta \big( Y \| Z \big) \le \epsilon$ if and only if there exists a random variable $Y'$ such that $\Delta(Y, Y') \le \delta$ and $\maxdiv^\delta \big( Y' \| Z \big) \le \epsilon$.
        \item We have both $\maxdiv^\delta \big( Y \| Z \big) \le \epsilon$ and $\maxdiv^\delta \big( Z \| Y \big) \le \epsilon$ if and only if there exist random variables $Y'$ and $Z'$ such that $\Delta(Y, Y') \le \delta / (e^\epsilon+1)$, $\Delta(Z, Z') \le \delta / (e^\epsilon+1)$, and $\maxdiv \big(Y' \| Z'\big) \le \epsilon$.
    \end{enumerate}
\end{lemma}

\begin{lemma}[e.g., Lemma 3.18 of Dwork and Roth \cite{DworkR/2014/FTTCS}]
    Suppose two random variables $Y$ and $Z$ satisfy that $\maxdiv \big( Y \| Z \big) \le \epsilon$ and $\maxdiv \big( Z \| Y \big) \le \epsilon$.
    Then, $\kl (Y \| Z) \le \epsilon \big( e^\epsilon - 1 \big)$
\end{lemma}

Consider any neighboring datasets $\mathcal{D}$ and $\widetilde{\mathcal{D}}$.
Let $P^t$ and $\widetilde{P}^t$ be random variables that are equal to the prices posted in round $t$ running Algorithm~\ref{alg:mw-optimized-step-sizes} on $\mathcal{D}$ and $\widetilde{\mathcal{D}}$ respectively.
Further, let $P^t(p^{1:t-1})$ denote the conditional random variable conditioned on $P^1 = p^1, P^2 = p^2, \dots, P^{t-1} = p^{t-1}$; 
define $\widetilde{P}^t(p^{1:t-1})$ similarly.

The following lemma follows from the discussions in Section~\ref{sec:privacy}.

\begin{lemma}
    For any $t \in [T]$, any $j \in [m]$, and any $p^1, p^2, \dots, p^{t-1}$, we have:
    \[
        \maxdiv^{\frac{\delta}{Tm}} \big( P^t_j(p^{1:t-1}) \| \widetilde{P}^t_j(p^{1:t-1}) \big) \le \epsilon^t
        ~,
    \]
    where:
    \[
        \epsilon^t = \frac{\epsilon \sqrt{\eta^t}}{\sqrt{m \etasum}}
    \]
\end{lemma}

Therefore, there exists $Q^t_j(p^{1:t-1})$ and $\widetilde{Q}^t_j(p^{1:t-1})$ such that:
\begin{align*}
    \Delta \big( P^t_j(p^{1:t-1}), Q^t_j(p^{1:t-1}) \big) & \le \frac{(\delta/Tm)}{1+\exp(\epsilon^t)} \le \frac{\delta}{2Tm} \\
    \Delta \big( \widetilde{P}^t_j(p^{1:t-1}), \widetilde{Q}^t_j(p^{1:t-1}) \big) & \le \frac{(\delta/Tm)}{1+\exp(\epsilon^t)} < \frac{\delta}{2Tm} \\[1ex]
    \maxdiv \big( Q^t_j(p^{1:t-1}) \| \widetilde{Q}^t_j(p^{1:t-1}) \big) & \le \epsilon^t
\end{align*}

The first two properties indicates that it suffices to show that:
\begin{equation}
    \label{eqn:coupled-variables-divergence}
    \maxdiv^\delta \big( Q^1, Q^2, \dots, Q^T \| \widetilde{Q}^1, \widetilde{Q}^2, \dots, \widetilde{Q}^T \big)
    \le \epsilon
\end{equation}

The third property, on the other hand, implies that:
\begin{align*}
    \kl \big( Q^t(p^{1:t-1}) \| \widetilde{Q}^t(p^{1:t-1}) \big) 
    & 
    = \sum_{j=1}^m \kl \big( Q^t_j(p^{1:t-1}) \| \widetilde{Q}^t_j(p^{1:t-1}) \big) \\
    &
    \le m \epsilon^t \big( \exp(\epsilon^t) - 1 \big) \\[2ex]
    &
    \le 2m  \big(\epsilon^t\big)^2
\end{align*}

Consider the following martingale, with $p^1, p^2, \dots p^T$ drawn from the distribution of $Q^1, Q^2, \dots, Q^T$:
\begin{align*}
    \ln \frac{\Pr \big[ Q^1 = p^1, \dots, Q^T = p^T \big]}{\Pr \big[ \widetilde{Q}^1 = p^1, \dots, \widetilde{Q}^T = p^T \big]} 
    &
    = \sum_{t=1}^T \ln \frac{\Pr \big[ Q^t(p^{1:t-1}) = p^t \big]}{\Pr \big[ \widetilde{Q}^t(p^{1:t-1}) = p^t \big]} 
\end{align*}

Its expectation is at most:
\[
    \sum_{t=1}^T 2m \big( \epsilon^t \big)^2 \le \frac{\epsilon^2}{\ln(2/\delta)}
\]

The sum of variance, in the sense of Theorem~\ref{thm:martingale-concentration}, is also bounded by this amount.
Then, by Theorem~\ref{thm:martingale-concentration}, we get that with probability at least $1-\delta$:
\[
    \ln \frac{\Pr \big[ Q^1 = p^1, \dots, Q^T = p^T \big]}{\Pr \big[ \widetilde{Q}^1 = p^1, \dots, \widetilde{Q}^T = p^T \big]} \le \epsilon
    ~.
\]

This proves Eqn.~\eqref{eqn:coupled-variables-divergence} and finishes the argument.

\end{document}